\documentclass[10pt]{article}

\usepackage{flushend, cuted}
\usepackage{amsmath}
\usepackage{amsfonts}
\usepackage{bbm}

\usepackage{amssymb}
\usepackage{float}
\usepackage{color,amsmath,amssymb, amsfonts,amstext,amsthm}

\usepackage{epsfig, graphicx, graphics}
\usepackage{longtable}
\usepackage{cite}
\usepackage{subfigure}

\renewcommand{\phi}{\varphi}

\newtheorem{theorem}{Theorem}

\title{Mean escape time for  randomly switching narrow gates in a cellular flow \footnote{
$\quad$$^{1}$ School of Mathematics and Statistics, Zhengzhou University, 100 Kexue Road, Zhengzhou 450001, China.
$\quad$$^{2}$ Department of Applied Mathematics, Illinois Institute of Technology, Chicago 60616, USA.
$\quad$$^{3}$ Center for Mathematical Sciences \& School of Mathematics and Statistics, Huazhong University of Science and Technology, Wuhan 430074,  China.}}

\author{Hui Wang$^{1}$, Jinqiao Duan$^{2,*}$, Xianguo Geng$^1$, Ying Chao$^3$}

\date{\today}

\begin{document}

\bibliographystyle{plain}

\maketitle

\begin{abstract}
The escape of particles through a narrow absorbing gate in confined domains is a abundant phenomenon in various systems in physics, chemistry and molecular biophysics.
We consider  the narrow escape problem  in a cellular flow when the two gates randomly switch between different states with a switching rate $k$ between the two gates. After briefly deriving the coupled   partial differential equations for the escape time through two gates, we compute the mean   escape time for particles escaping from the gates with different initial states. By  numerical simulation under nonuniform boundary conditions, we quantify how  narrow  escape time is affected by the switching rate $k$ between the two gates, arc length $s$  between two gates, angular velocity $w$ of the cellular flow and diffusion coefficient  $D$.  We reveal that the mean escape time decreases with the switching rate $k$ between the two gates, angular velocity $w$ and diffusion coefficient $D$ for fixed arc length, but takes the minimum when the two gates are evenly separated on the  boundary for any given switching rate $k$ between the two gates. In particular, we find that  when  the    arc length size $\varepsilon$ for the gates  is sufficiently small, the average narrow  escape time  is approximately independent of  the gate arc length size.  We further indicate  combinations of  system parameters (regions located  in the parameter space) such that the mean escape time is the longest or shortest. Our findings provide  mathematical understanding for  phenomena such as how ions  select ion channels and how       chemicals leak in annulus ring containers, when drift vector fields are present.

\end{abstract}

\section{Introduction}

The narrow escape problems have attracted a lot of attention in recent years, due to its significant relevance in cellular and molecular biology,  and chemical physics \cite{Holcman2015, Bressloff2014, Bressloff2013, Holcman2017, Holcman2012, Reingruber2009a}. The heart of a narrow escape problem is to compute the narrow escape time  (NET), namely the mean time for a particle starting in the domain before   exiting through a narrow gate on the boundary.

In the biophysics context, when an ion searches for an open ion channel within the membrane, neurotransmitter receptors transmit signal within a synapse of a neuron, a reactive particle to activate a given protein, or  newly transcribed mRNA transports from the nucleus to the cytoplasm through nuclear pores \cite{Benichou2008, Schuss2007, Holcman2004}. These examples  have something in common: `Particles' are all confined to a bounded domain with small exits (or holes, or gates) on the boundary of the domain. Meanwhile, in order to fulfill their biological function, particles must exit from the small gates  on the boundary. For instance, viruses must enter the cell nucleus through small nanopores in order to replicate \cite{Lagache2017}.

The narrow escape problems have been studied via various stochastic models recently. Most of the scenarios considered are two-dimensional \cite{Ammari2011, Grebenkov2016, Li2014, Lindsay2015, Pillay2010}, and some are one-dimensional \cite{Reingruber2010, Reingruber2009b} or three-dimensional \cite{Reingruber2009b, Gomez2015, Li2017}. In addition, some authors  consider the boundary with one gate \cite{Benichou2008,Chevalier2011, Singer2008, Agranov2018}, with two gates\cite{ Reingruber2010, Chevalier2011}, while others multiple gates \cite{Lagache2017, Chevalier2011, Cheviakov2012}.

 Several papers \cite{Lindsay2015, Reingruber2009b, Gomez2015, Holcman2013, Holcman2008} have been devoted to the asymptotic approximate expressions  for the narrow escape time.
 Since the mixed boundary conditions, there is no exact solution for a narrow escape problem. The narrow escape problems are mostly studied analytically \cite{Grebenkov2016,  Lindsay2015, Pillay2010, Reingruber2010, Gomez2015, Chevalier2011, Singer2008, Cheviakov2012,Grebenkov2017, Benichou2015, Bressloff2015, Piazza2015,  Singer2007, Berezhkovskii2010, Levernier2017}. For example, a general exact formula for the mean first passage time from a fixed point inside a planar domain to an escape region on its boundary is attained in \cite{Grebenkov2016}; high-order asymptotic formulas for the mean first passage time  are derived in \cite{Lindsay2015}; approximations for the average mean first passage time  are found  in\cite{Gomez2015}; the upper and lower bounds of the mean first passage time  for mortal walkers are derived  in\cite{Grebenkov2017}. Moreover, a spectral approach to derive an exact formula for the mean exit time of a particle through a hole on the boundary is developed    in\cite{Benichou2015};   partial differential equation and probabilistic methods are applied to solve escape problems in \cite{Bressloff2015};   a generalized Kramers formula for the mean escape time through a narrow window is obtained in\cite{Singer2007}; and boundary homogenization is used when the   boundary contains  non-overlapping identical absorbing arcs in \cite{Berezhkovskii2010}.

 A few authors numerically studied the narrow escape problems, including    Brownian simulations  \cite{Reingruber2009b} and  Monte Carlo simulations  \cite{Rojo2012}. Think of the narrow escape problem the other way around, Agranov and Meerson  \cite{Agranov2018} developed a formalism to evaluate the nonescape probability of a gas of diffusing particles.

Most works approximately estimate the narrow escape time of particles under Brownian diffusion (i.e., no drift or no external vector field). It turns out that the escape time is described by a mixed Dirichlet-Neumann boundary value problem for the Laplace operator in a bounded domain \cite{Holcman2015, Lions1984}. Very recently, Lagache and Holcman
 \cite{Lagache2017}  considered  the narrow escape problem for particles in a constant drift, in order  to analyze viral entry into the cell nucleus.


In this paper, we consider a two dimensional model of narrow escape problem for particles  in a unit annulus,  with two randomly switching gates on the outer boundary, under a   vector field (i.e., a cellular flow). The circular annulus may represent the cell cytoplasm, with outer radius $r_2$ and inner radius $r_1$. The two dimensional model may represent flat culture cells. Since the adhesion to the substrate, they stay flat. The thickness can be neglected in this idealized model.


This paper is organized as follows. In Section 2, we first describe a stochastic narrow escape model, then  briefly drive the coupled partial differential equations to be satisfied by the mean escape time through two gates. In Section 3, we present numerical experiments and analysis.  We end the paper with a discussion  in Section 4.

\section{Materials and methods}

\subsection{Stochastic narrow escape model}
\begin{figure}[!htb]
\centering\includegraphics[height=6cm ,width=6cm]{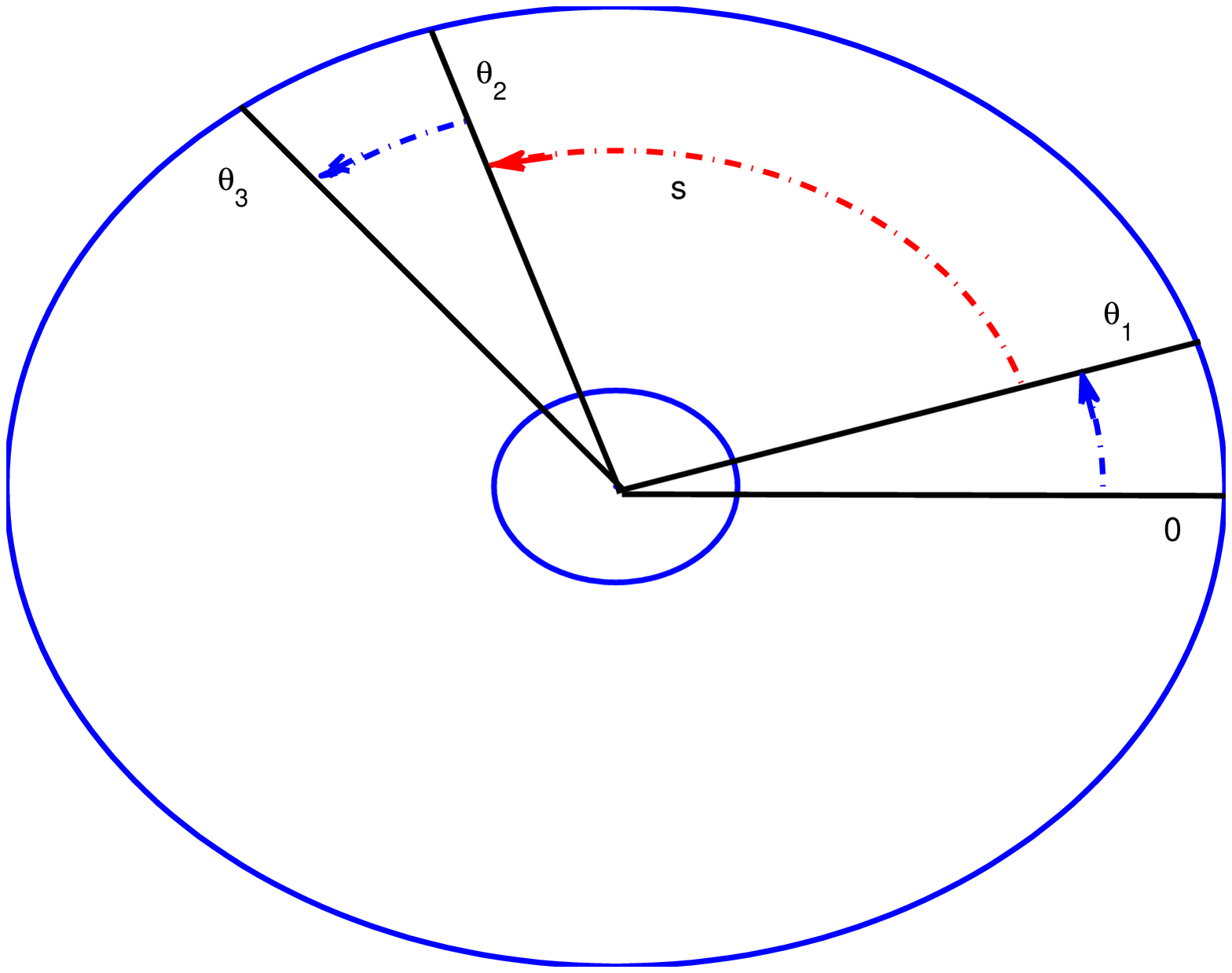}
\caption{(Color online) A  unit annular domain with small absorbing gates on an otherwise reflecting outer circular boundary.}
\label{Fig.0}
\end{figure}

We consider particles moving in an annular domain with  inner radius $r_1$ and outer radius $r_2$, in  a   cellular  flow.
The   steady cellular  flow \cite{Acheson1990} is  induced by the rotation of the inner boundary with angular velocity $w$, and is expressed in polar coordinates as $(0,  A r+\frac {B}{r})$
with
\begin{eqnarray}\label{coeffient}
 A =  \frac { -w r_1^2}{r_2^2-r_1^2}, \\
 B =  \frac {w r_1^2 r_2^2}{r_2^2-r_1^2}.
\end{eqnarray}
By the transformation of polar coordinates, $x=r \cos \theta$, $y=r \sin \theta$ with  $r=\sqrt{x^2+y^2}$,  the cellular flow in the  Cartesian  coordinates  is \\
\begin{eqnarray}\label{drift}
 b_1(x, y)= -y \frac {-w r_1^2 (x^2+y^2)+w r_1^2 r_2^2 }{\sqrt{x^2+y^2}(r_2^2-r_1^2)}, \\
 b_2(x, y) = x  \frac {-w r_1^2 (x^2+y^2)+w r_1^2 r_2^2 }{\sqrt{x^2+y^2}(r_2^2-r_1^2)}.
\end{eqnarray}
With the driving cellular flow $b(x, y):=(b_1(x, y), b_2(x, y))$, we consider the two dimensional domain $\Omega$,  whose boundary  is decomposed as $\partial \Omega=\partial \Omega_c \bigcup \partial \Omega_1 \bigcup\partial \Omega_2 $.   The boundary has two  small gates at the arcs $\partial \Omega_1$ and $\partial \Omega_2$, together with the remaining part $\Omega_c$.
A  particle's  trajectory  is  reflected on   $\Omega_c$,  but  absorbed at  gates $\Omega_1$ or $\Omega_2$.  Each gate is a small absorbing arc. We consider the  following stochastic narrow escape model, i.e.,  the equations of motion for  a particle
\begin{align}
\begin{cases}
   d x_t &=b_1(x_t, y_t)dt+\sqrt {2D} d B_t^1- n_1(x_t, y_t)dL_t,\\
   d y_t &=b_2(x_t, y_t)dt+\sqrt {2D} d B_t^2-  n_2(x_t, y_t)dL_t,
   \end{cases}
\end{align}
where $ (b_1(x, y), b_2(x, y))$ is a drift vector field (a cellular flow),  $(B_t)_{t\geq 0} = (B_t^1,  B_t^2)$  is a standard Brownian motion, $n = ( n_1, n_2 )$ is the outward unit normal vector on the boundary, $(L_t)_{t\geq 0}$ is a continuous nondecreasing process (with $L_0=0$) which increases only when $(x, y)$ is on the boundary $\partial \Omega $:
\begin{equation}
L_t=\int_0^t\mathbb{I}_{(x, y)\in \partial \Omega}dL_s.
\end{equation}

Assume that the  two gates at the arcs $\partial \Omega_1$ and $\partial \Omega_2$   open alternatively \cite{Ammari2011, Pillay2010} in terms of a telegraph process $(N_t)_{t\geq0}$ with parameter $k$, where  $k$ is the rate of switching between the two gates. When the particle in the cytoplasm hits the open gate it escapes, if not it will be reflected. The telegraph process $(N_t)_{t\geq0}$ with parameter $k$ is a Markov process which takes on only two values, $1$ and $2$. The different values of Markov process corresponding to different states, different function of cell \cite{Bardet2009}. The process can be described as follows
$$
N_t=
\left \{
  \begin{array}{ll}
    1, &  \hbox{if $ T_{2j}\leq t< T_{2j+1}$ for some $j$,} \\
    2, & \hbox{if $ T_{2j+1}\leq t< T_{2j+2}$ for some $j$,}
  \end{array}
\right.
$$
where $T_0=0$, $T_j=\sum_{l=1}^j\tau_l$, and $\tau_l$   is a sequence of independent and identically distributed exponential random variables with parameter $k>0$. That is, $\mathbb{E}[\tau_l]=\frac{1}{k}$.
Denote $\textbf{X}_t =(x_t, y_t)$ and $\textbf{x} =(x, y)$.
We denote $\mathcal{T}$ as the stopping time corresponding to the escape of the particle:
$\mathcal{T}=\mathcal{T}_1\wedge\mathcal{T}_2$, $\mathcal{T}_j=\inf \{t\geq0, N_t=j \ and \ \mathbf{X}_t\in \partial \Omega_j\}$, j=1, 2.

%

\subsection{ Narrow escape time with drift }
Our goal is to compute the expected quantities $u_1(\textbf{x}) = \mathbb{E}_{1,\textbf{x}} [\mathcal{T}]$ and
$u_2(\textbf{x}) = \mathbb{E}_{2,\textbf{x}} [\mathcal{T}]$  of the stopping time given the initial state is $N_0=j$ and $\mathbf{X}_0=\mathbf{x}$, j=1, 2.

More specifically, $u_1(x, y)$ is the mean time that the particle starting at  $(x, y)$ in the domain escape from any gate with the first gate open initially. Similarly,
$u_2(x, y)$ is the mean time that the particle starting at  $(x, y)$  in the domain escape from any gate with the second gate open initially. We  now prove the following theorem.

\begin{theorem}
The  mean  narrow escape times $u_1, u_2$  satisfy the following coupled partial differential equations \\
\begin{equation}\label{meannarrow}
   -\begin{pmatrix}
   1 \\
   1\end{pmatrix} =
   k\begin{pmatrix}
   -1 & 1\\
   1  & -1
   \end{pmatrix}
   \begin{pmatrix}
   u_1(x, y) \\
   u_2(x, y)
   \end{pmatrix} +
   D \Delta\begin{pmatrix}
   u_1(x, y) \\
   u_2(x, y)
   \end{pmatrix} +
   \begin{pmatrix}
   b  \cdot  \nabla u_1(x, y)  \\
   b  \cdot  \nabla  u_2(x, y)
   \end{pmatrix}
\end{equation}
with nonuniform boundary conditions:
\begin{eqnarray}\label{Neu1}
  \partial_n u_j |_{\partial \Omega_c} = 0, j=1, 2,
\end{eqnarray}

\begin{eqnarray}\label{Dir}
 u_1 |_{\partial \Omega_1} = 0,  u_2 |_{\partial \Omega_2} = 0,
\end{eqnarray}

\begin{eqnarray}\label{Neu2}
  \partial_n u_1 |_{\partial \Omega_2} = 0,   \partial_n u_2 |_{\partial \Omega_1} = 0,
\end{eqnarray}
where $\partial_n $ represents the normal derivative and $n$ is outward unit normal vector on the outer boundary.
\end{theorem}
\begin{proof}
Note that the right-hand side of the (\ref{meannarrow}) is the infinitesimal generator $\mathcal{A}$ of the process $(\mathbf{X}_t, N_t)_{t\geq 0}$. Denoting $u(\mathbf{x}, j)=u_j(\mathbf{x})$, by It\^{o} formula \cite{DuanBook2015}, we have the following relation, for   $s<t$,
\begin{eqnarray}
& & \mathbb{E}[u(\mathbf{X}_{t \wedge \mathcal{T}}, N_{t \wedge \mathcal{T}})-u(\mathbf{X}_{s \wedge \mathcal{T}}, N_{s \wedge \mathcal{T}})-\int_{s \wedge \mathcal{T}}^{t \wedge \mathcal{T}}
\mathcal{A} u(\mathbf{X}_r, N_r)dr|\mathcal{F}_{s \wedge \mathcal{T}}]   \nonumber \\
&=&-\mathbb{E}[\int_{s \wedge \mathcal{T}}^{t \wedge \mathcal{T}} \partial_\nu u(\mathbf{X}_r, N_r)dL_r|\mathcal{F}_{s \wedge \mathcal{T}}]+\mathbb{E}[\int_{s \wedge \mathcal{T}}^{t \wedge \mathcal{T}} \nabla u(\mathbf{X}_r, N_r)\sqrt {2D}dB_r|\mathcal{F}_{s \wedge \mathcal{T}}]    \nonumber \\
&=&-\mathbb{E}[\int_{s \wedge \mathcal{T}}^{t \wedge \mathcal{T}} \partial_\nu u(\mathbf{X}_r, N_r)dL_r|\mathcal{F}_{s \wedge \mathcal{T}}]  \label{eq1}
\end{eqnarray}
where $(\mathcal{F}_t)_{t\geq 0}$ is the filtration of the process $(\mathbf{X}_t, N_t)_{t\geq 0}$, and $L_t$ is defined as above. Neumann boundary conditions (\ref{Neu1}) and (\ref{Neu2}) show us that $\partial_\nu u(\mathbf{X}_r, N_r)=0$ if $r<  \mathcal{T}$. By the given condition, we know $\mathcal{A} u(\mathbf{X}_r, N_r)=-1$. This   leads to
\begin{eqnarray}
\mathbb{E}[u(\mathbf{X}_{t \wedge \mathcal{T}}, N_{t \wedge \mathcal{T}})-u(\mathbf{X}_{s \wedge \mathcal{T}}, N_{s \wedge \mathcal{T}})+(t \wedge \mathcal{T}-s \wedge \mathcal{T})|\mathcal{F}_{s \wedge \mathcal{T}}]=0,     \nonumber
\end{eqnarray}
which means the process $(u(\mathbf{X}_{t \wedge \mathcal{T}}, N_{t \wedge \mathcal{T}})+t \wedge \mathcal{T})_{t\geq 0}$ is a martingale. Then applying this equality with $s=0$ and letting $t\rightarrow \infty$, we get
\begin{eqnarray}
\mathbb{E}[u(\mathbf{X}_{\mathcal{T}}, N_{\mathcal{T}})-u(\mathbf{X}_0, N_0)+ \mathcal{T}|\mathcal{F}_0]=0. \nonumber
\end{eqnarray}
From the Dirichlet boundary conditions in (\ref{Dir}), we have $u(\mathbf{X}_{\mathcal{T}}, N_{\mathcal{T}})=0$. Thus we obtain
\begin{eqnarray} \label{equproof}
\mathbb{E}[\mathcal{T}|\mathcal{F}_0]=\mathbb{E} [u(\mathbf{X}_0, N_0)|\mathcal{F}_0].
\end{eqnarray}
Finally, we take the expectation of both sides of the equation (\ref{equproof}),
\begin{eqnarray}
\mathbb{E}_{k,\mathbf{x}}\mathcal{T}=\mathbb{E} u(\mathbf{x}_0, N_0)=\mathbb{E}u(\mathbf{x}, k)=u_k(\mathbf{x}). \nonumber
\end{eqnarray}
This gives the desired result when the initial distribution is such that  $N_0=k$ and $\mathbf{X}_0=\mathbf{x}$.
\end{proof}
Note that Ammari et al. \cite{Ammari2011} showed this theorem in the case of no drift.

In the polar coordinates, the NET $u(r, \theta)$ is the solution of the preceding  boundary value problem given the initial point $(r, \theta)$. It is a reflected diffusion process to the absorbing boundary $\partial \Omega_j, (j=1,2)$.
In the context of polar coordinates: $x=r \cos\theta$, $y=r \sin \theta$,
we obtain that
\begin{eqnarray}
  \partial_\nu u= \partial_x u \cos\theta+\partial_y u \sin \theta=0,
\end{eqnarray}

\begin{eqnarray}
  \partial_x u= \partial_r u \partial_x r+\partial_\theta u \partial_x \theta=\cos\theta \partial_r u-\frac{\sin \theta}{r} \partial_\theta u,
\end{eqnarray}

\begin{eqnarray}
  \partial_y u= \partial_r u \partial_y r+\partial_\theta u \partial_y \theta=\sin \theta \partial_r u+\frac{\cos\theta}{r} \partial_\theta u,
\end{eqnarray}
and
\begin{eqnarray}
  \Delta u=\partial_{xx} u +\partial_{yy} u =\partial_{rr} u+\frac{1}{r^2}\partial_{\theta\theta} u+\frac{1}{r}\partial_r u.
\end{eqnarray}

We use a finite difference method to solve the Dirichlet-Neumann boundary value problem (\ref{meannarrow})-(\ref{Neu2}) in the polar coordinates.

\section{Results and Discussion}
We now describe our   numerical simulation results of the narrow escape time (NET). The Dirichlet-Neumann boundary value problem in Theorem 1 is solved by a finite difference method. The goal here is to present the results for the  averaged NET from the numerical simulation, then conduct    comparisons for different parameters.

The average NET is the average   escape  time that all the particles in the annular domain take to escape the annular domain from any gate. It is also called average mean first passage time  \cite{Cheviakov2010, Chen2011}. In the context of polar coordinates, the average NET is given by
\begin{eqnarray}
  \frac{1}{\pi(r_2^2-r_1^2)}\int\int_{\Omega'}u(r,\theta)drd\theta,
\end{eqnarray}
where $\Omega'=[r_1, r_2]\times[0, 2\pi]$.

\begin{figure}[!htb]
\subfigure[]{ \label{Fig.sub.31}
\includegraphics[width=0.45\textwidth]{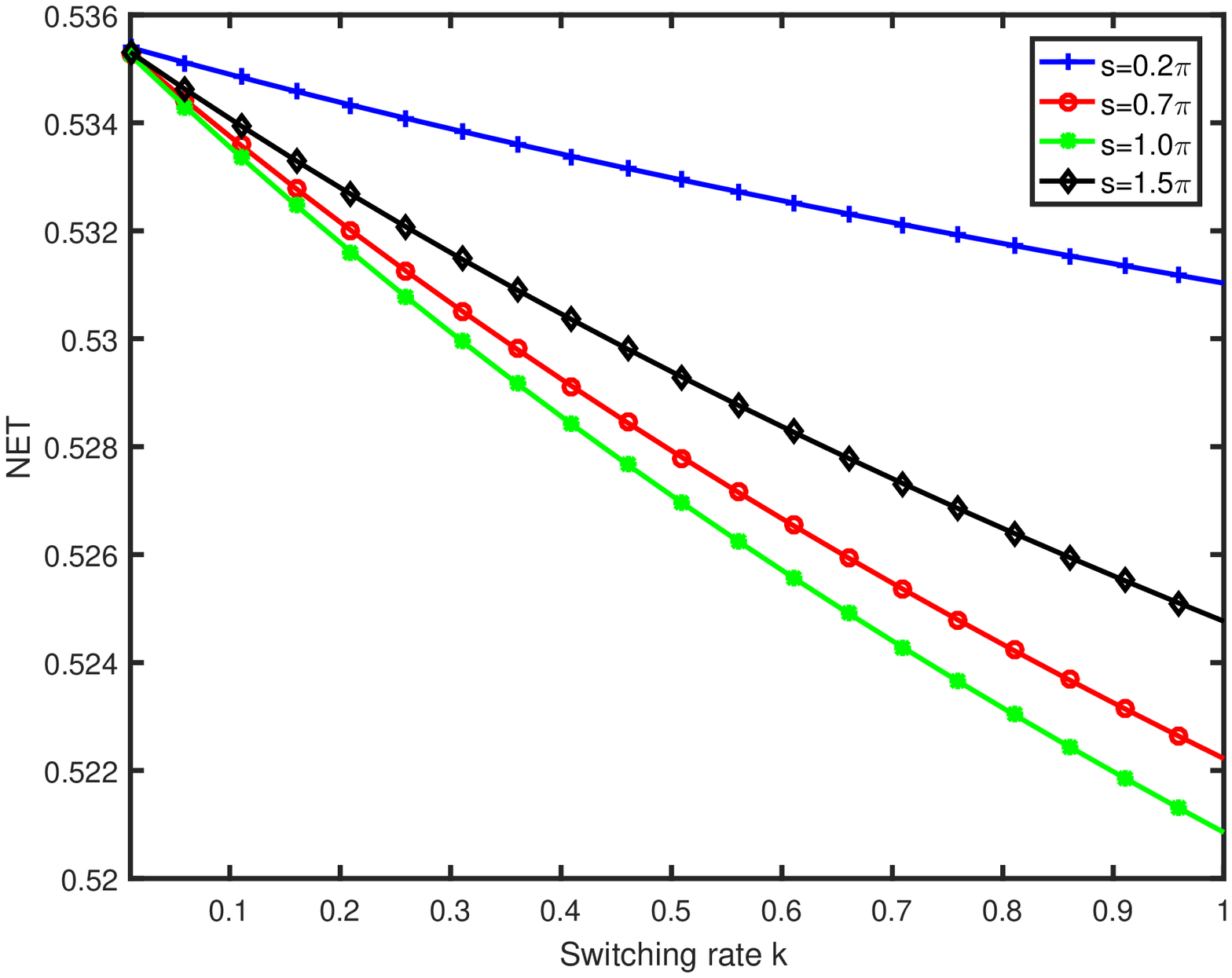}}
\subfigure[]{ \label{Fig.sub.32}
\includegraphics[width=0.45\textwidth]{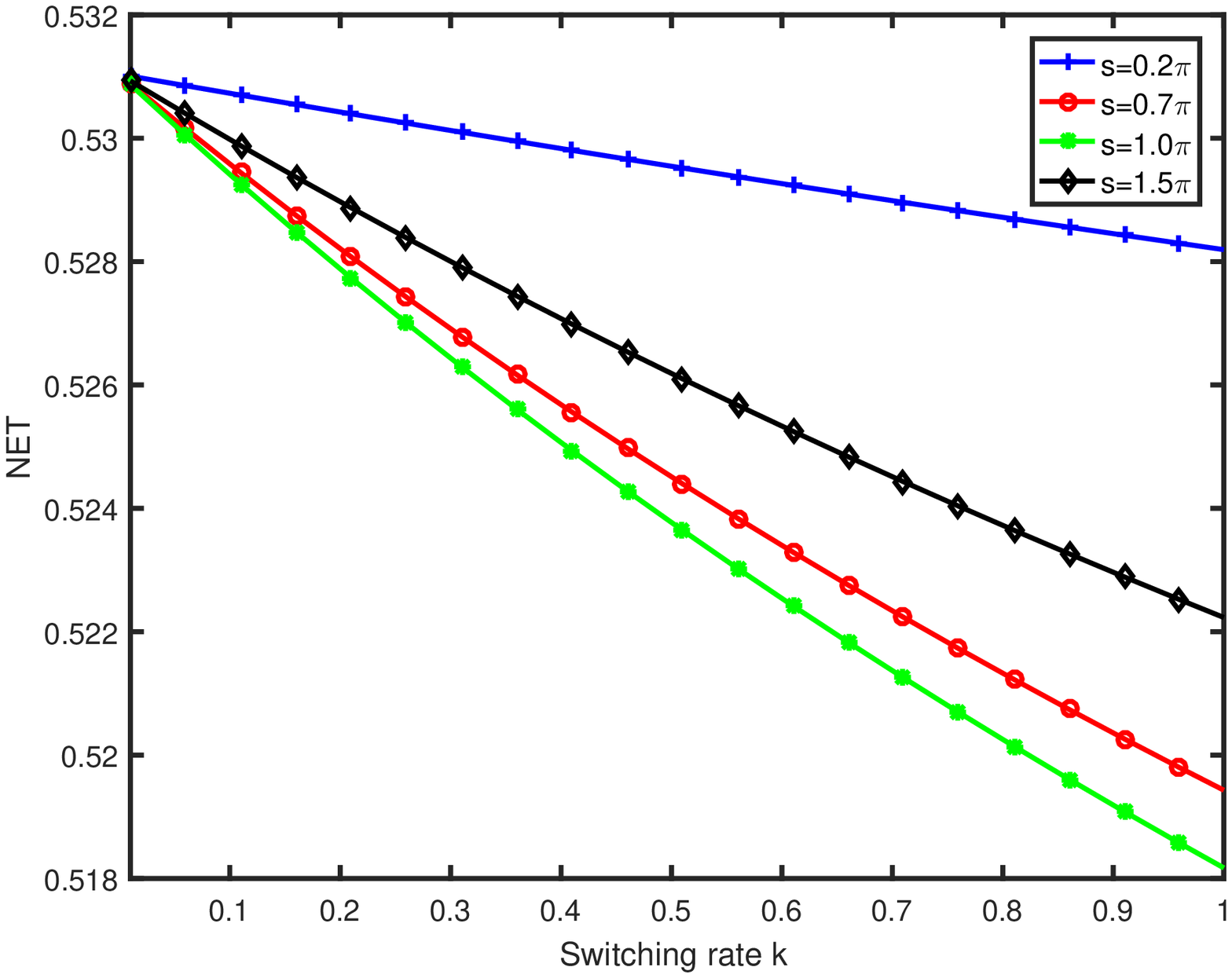}}
\caption{(Color online) Average narrow escape time with switching rate $k$ between the two gates for diffusion coefficient $D=1$ and angular velocity $w=3$: (a) With the first gate open initially. (b) With the second gate open initially.}
 \label{Fig_3}
\end{figure}

\begin{figure}[!htb]
\subfigure[]{ \label{Fig.sub.6111}
\includegraphics[width=0.45\textwidth]{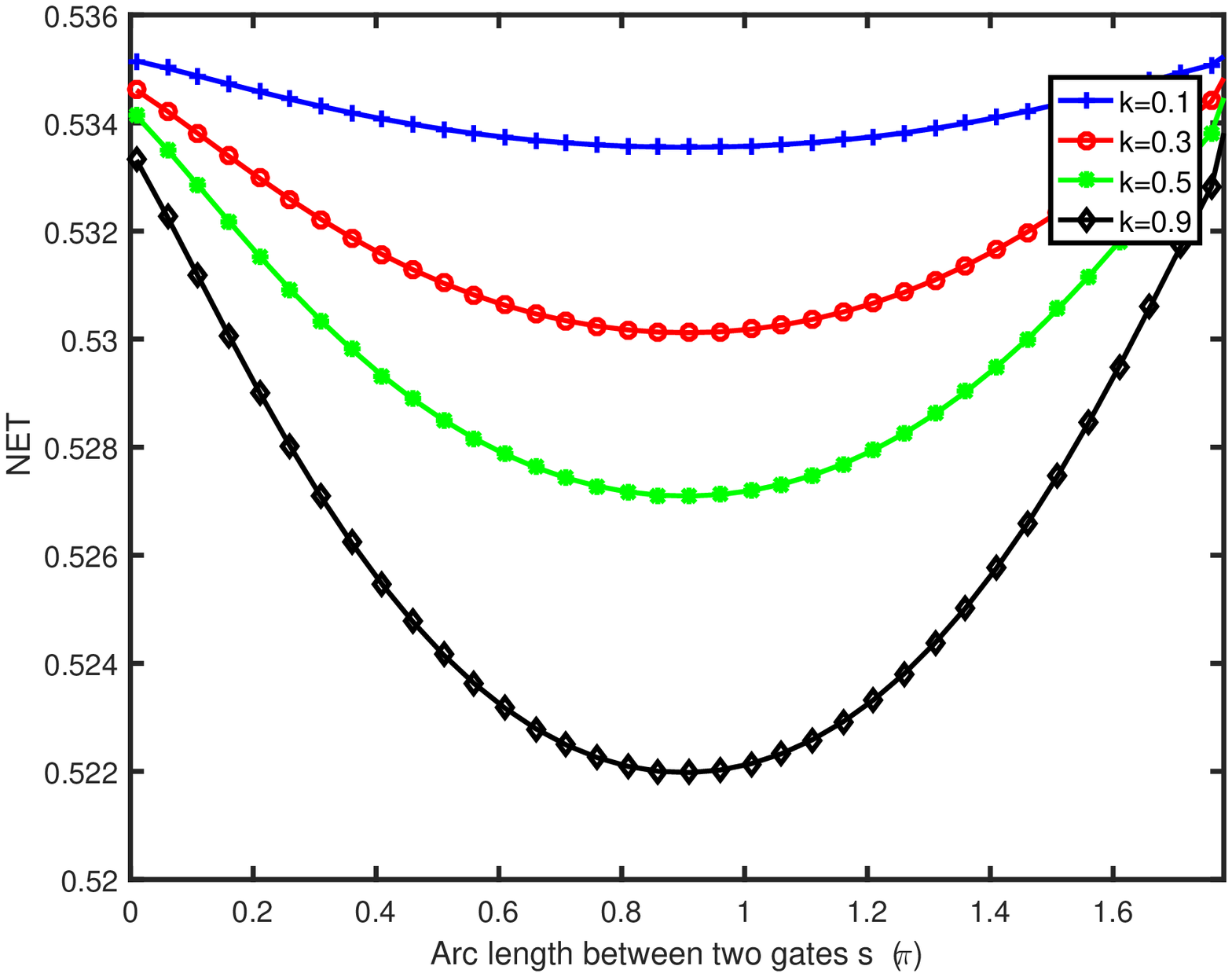}}
\subfigure[]{ \label{Fig.sub.6211}
\includegraphics[width=0.45\textwidth]{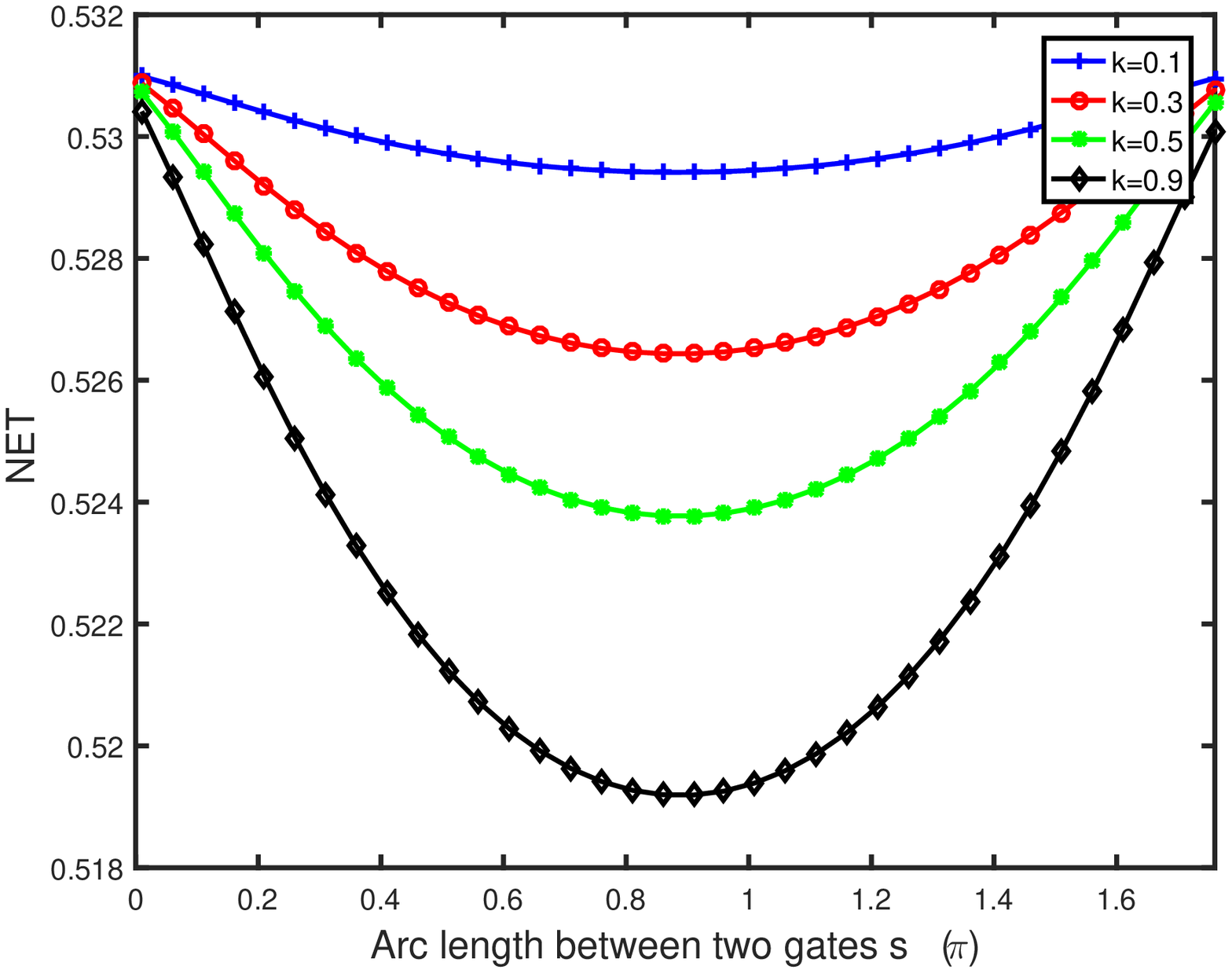}}
\caption{(Color online) Average narrow escape time with the arc length $s$ between two gates for diffusion coefficient $D=1$ and angular velocity $w=3$: (a) With the first gate open initially. (b) With the second gate open initially.}
 \label{Fig_611}
\end{figure}

First, we show the relationship of average NET with switching rate $k$ between the two gates and   arc length (`distance')  $s$ between two gates in Fig \ref{Fig_3} and Fig \ref{Fig_611}.

In Fig \ref{Fig_3}, the average NET is plotted as a function of switching rate between the two gates $k$. We compare  average NET for  different arc length $s$ between two gates. when the arc length between two gates $0<s\leq \pi$, the average NET decreases with the increase of arc length $s$ between two gates , while for $\pi<s< 2\pi$, the  average NET  is longer than that of $s=\pi$. To understand the reason for this at $s=\pi$,  note that the two gates are evenly distributed on the cell boundary, so the particle may choose the nearest side to escape quickly; but for $0<s<\pi$ or $\pi<s< 2\pi$, the two gates are on the same side of the semicircle, thus the particle in the other  semicircle has to cross over to the other side then escape. Fig \ref{Fig_611} plots the effects of arc length $s$ between two gates on average NET for various switching rate $k$ between the two gates. As we explained above, the average NET is smallest when  arc length $s$ between two gates is around $\pi$. Meanwhile, Fig \ref{Fig_3} and Fig \ref{Fig_611} both suggest  that the average NET decrease as the switching rate $k$ between the two gates  increases, this is consistent  with   known theoretical result \cite{Ammari2011}.


\begin{figure}[!htb]
\subfigure[]{ \label{Fig.sub.61}
\includegraphics[width=0.45\textwidth]{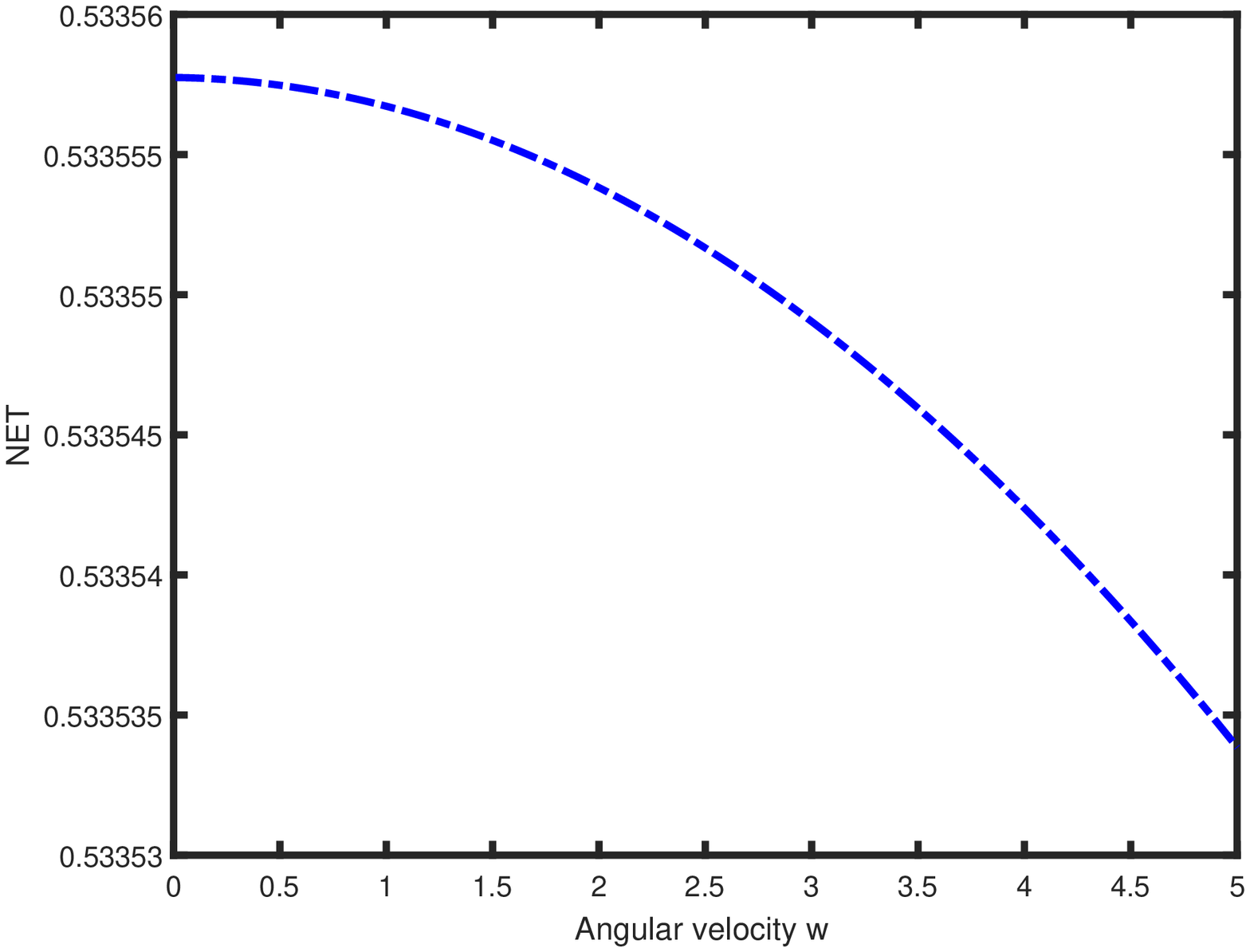}}
\subfigure[]{ \label{Fig.sub.62}
\includegraphics[width=0.45\textwidth]{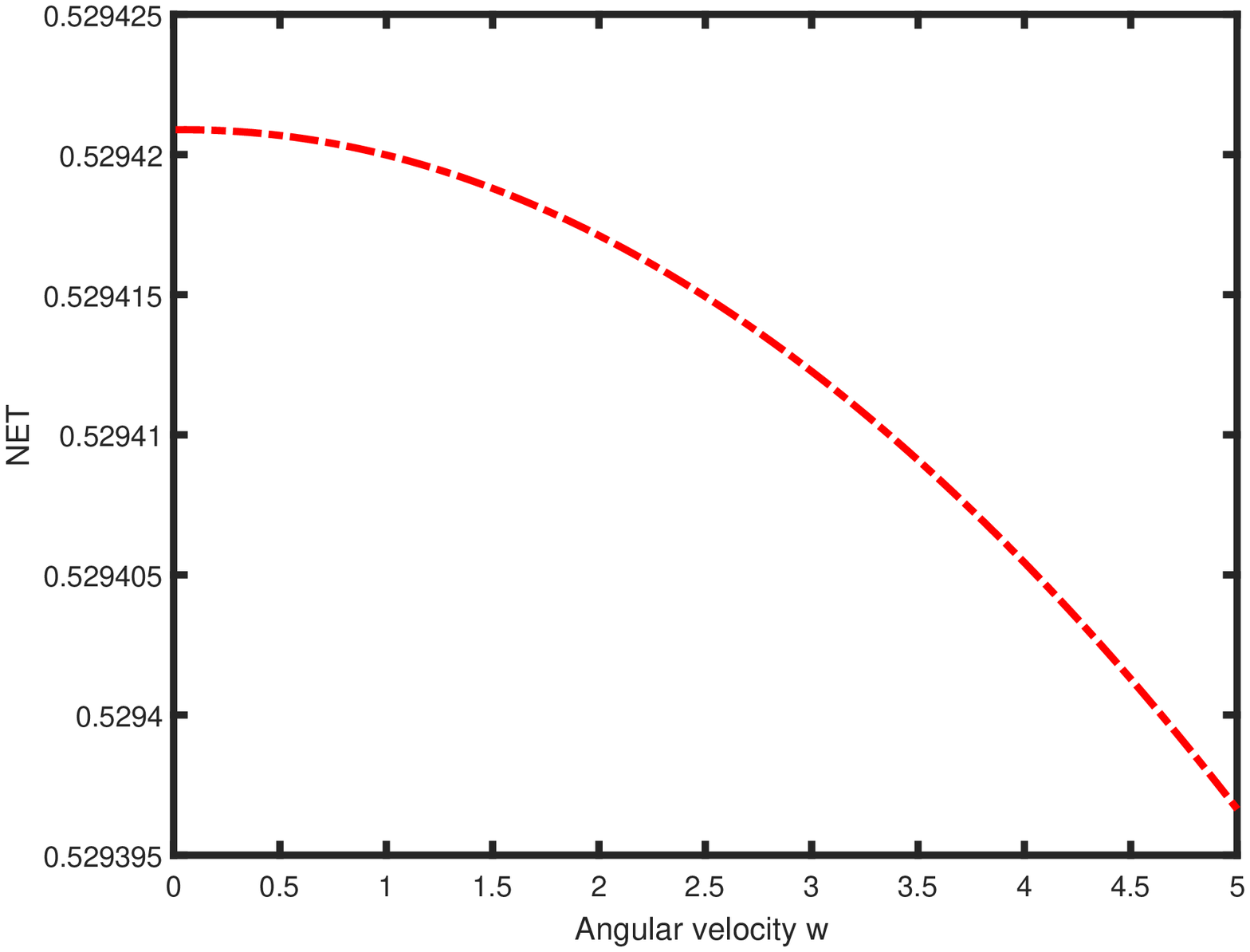}}
\caption{(Color online)Average narrow escape time with  angular velocity $w$ for switching rate $k=0.1$ between  the two gates, diffusion coefficient $D=1$ and  arc length $s=\pi$ between two gates:  (a)With the first gate open initially. (b) With the second gate open initially.}
 \label{Fig_601}
\end{figure}


Each curve in Fig \ref{Fig_601} predicts an decrease in the average NET with angular velocity $w$. We plot here for switching rate $k=0.1$ between the two gates, for $0<k<1$, it has similar trends.

\begin{figure}[!htb]
\subfigure[]{ \label{Fig.sub.6113}
\includegraphics[width=0.45\textwidth]{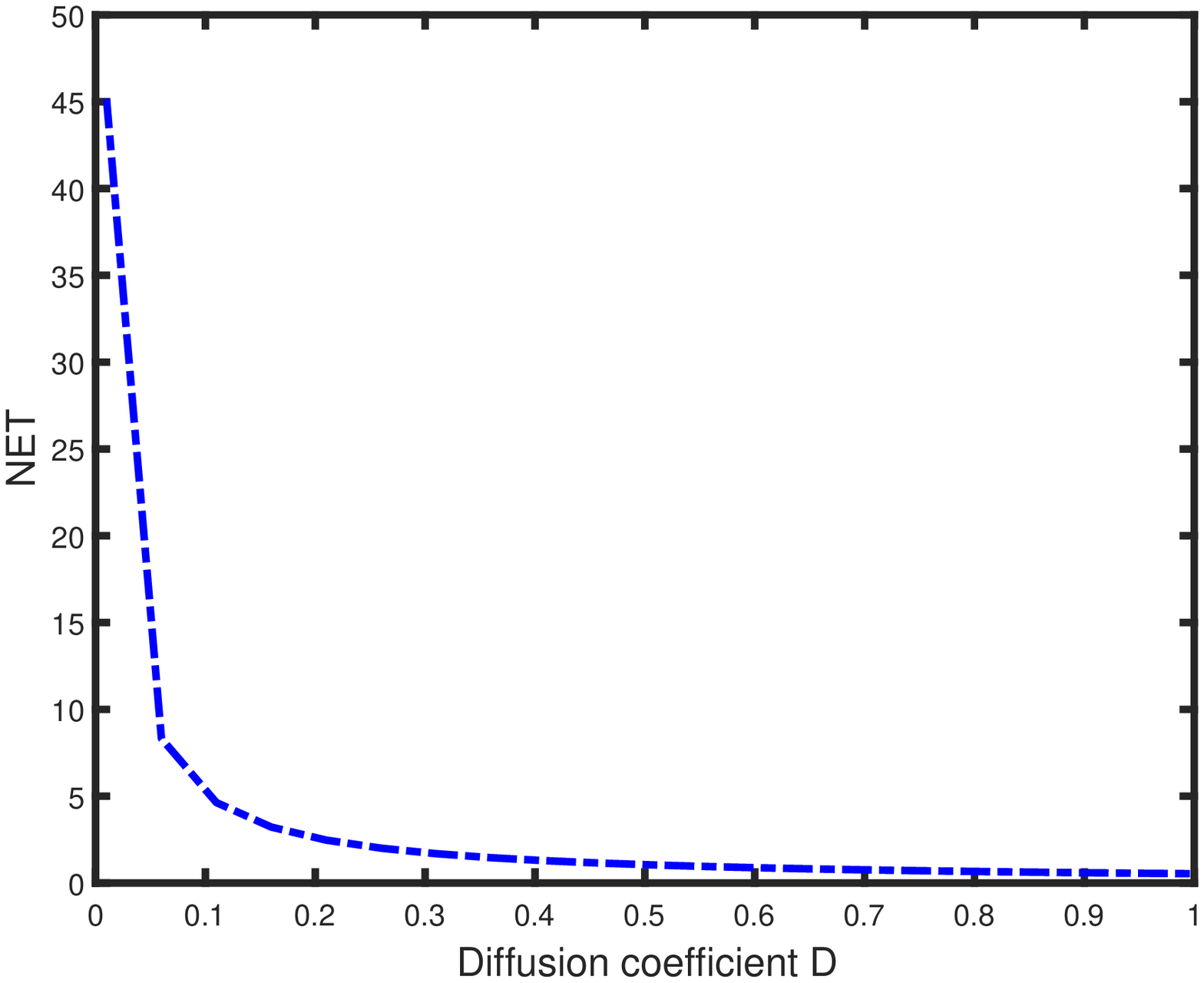}}
\subfigure[]{ \label{Fig.sub.6213}
\includegraphics[width=0.45\textwidth]{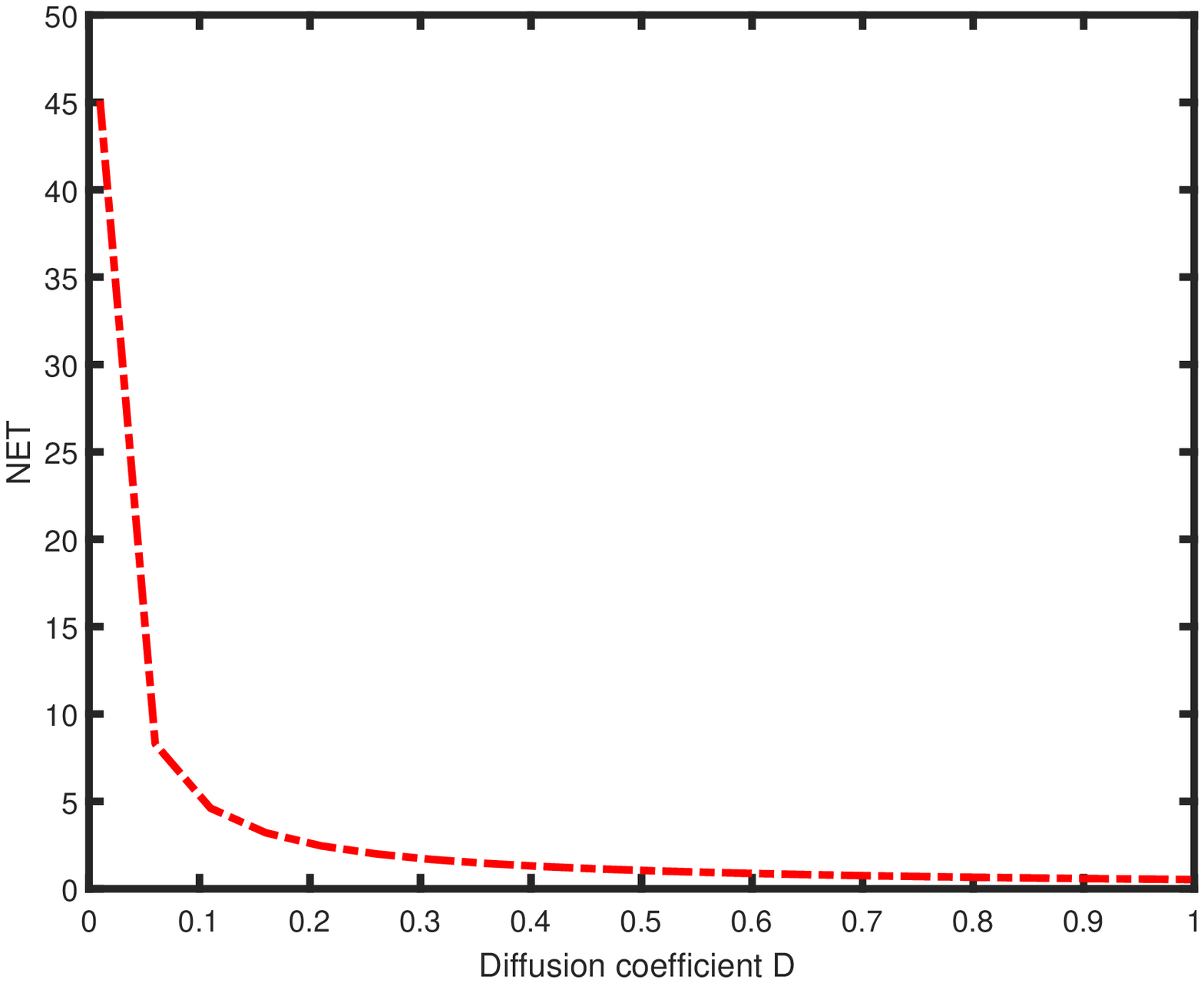}}
\caption{(Color online) Average narrow escape time   with diffusion coefficient  $D$ for switching rate $k=0.3$ between  the two gates, angular velocity $w=3$ and arc length $s=\pi$ between two gates:  (a) With the first gate open initially. (b) With the second gate open initially.}
 \label{Fig_613}
\end{figure}

Fig \ref{Fig_613} demonstrates that the effects of diffusion coefficient  $D$ on the average NET. We   see that the average NET has a monotonic behavior with the diffusion coefficient $D$. The average NET   declines when  $D$ is small, then when average NET crosses the inflection point, it changes more modestly. In this case, small diffusion coefficient $D$ is not good for particles to escape quickly.

\begin{figure}[!htb]
\subfigure[]{ \label{Fig.sub.71}
\includegraphics[width=0.45\textwidth]{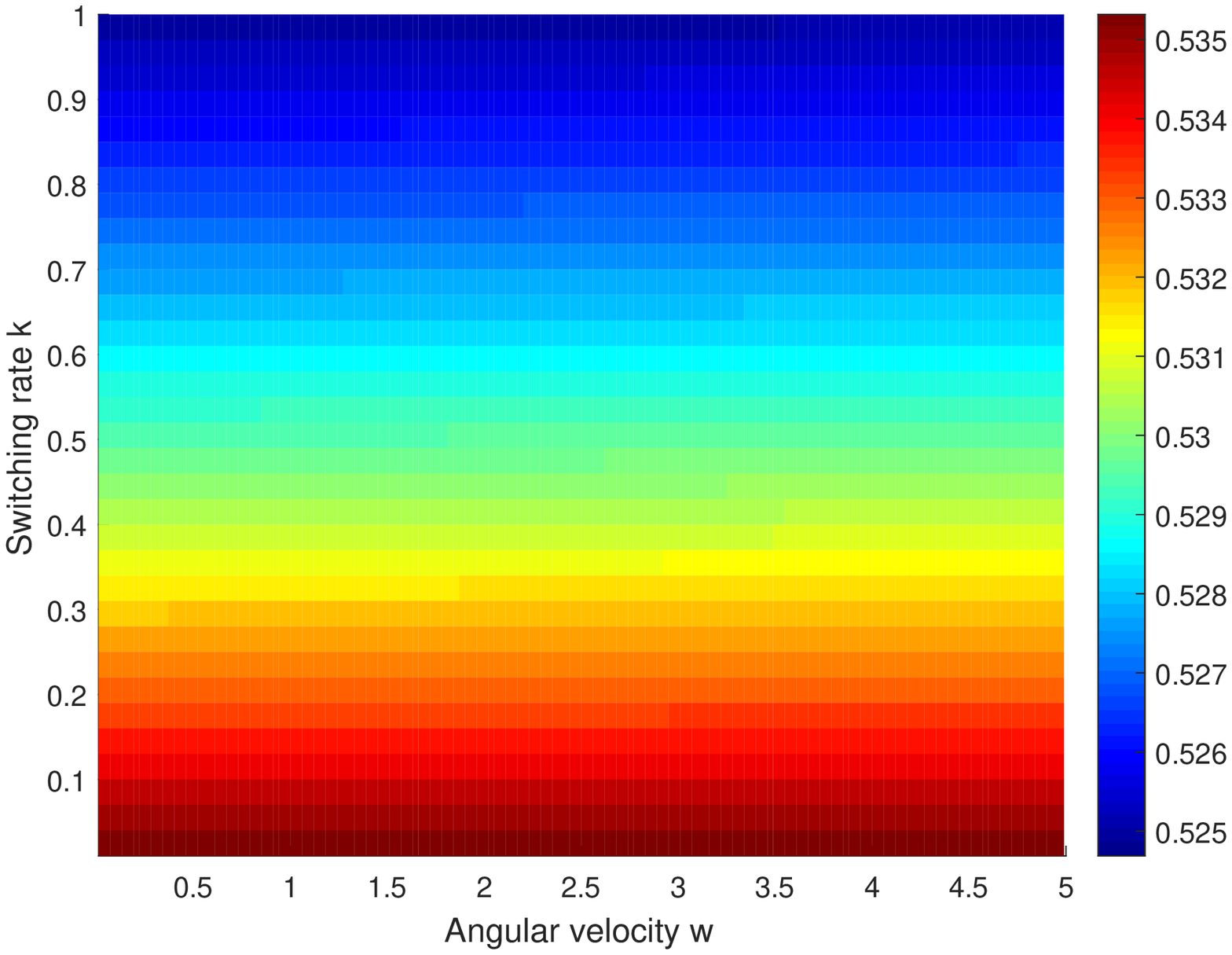}}
\subfigure[]{ \label{Fig.sub.72}
\includegraphics[width=0.45\textwidth]{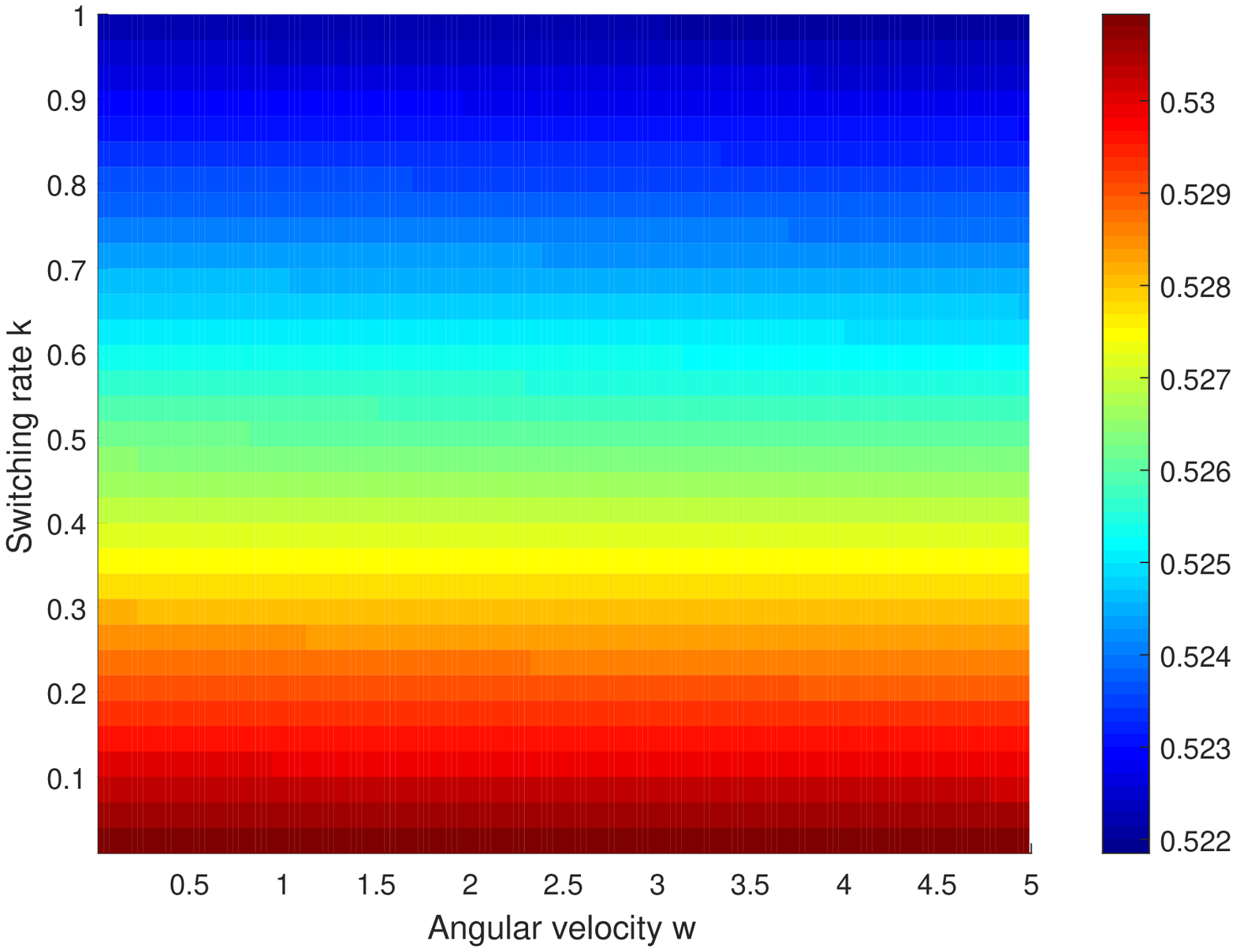}}
\caption{(Color online) Average narrow escape time with angular velocity $w$ and switching rate $k$ between the two gates for diffusion coefficient $D=1$ and  arc length $s=0.5 \pi$ between two gates: (a) With the first gate open  initially. (b) With the second gate open  initially.}
 \label{Fig_7}
\end{figure}

Fig \ref{Fig_7} presents the combined effects of angular velocity $w$ and switching rate $k$ between the two gates  on average NET. In  both Fig \ref{Fig_7}(a) and Fig \ref{Fig_7}(b), the red region presents the larger average NET domain, while the blue region represents the smaller parts. The smaller average NET values occur when $(w, k)$ is in the  blue domains at the top. Different initial states have very little influence on average NET.

\begin{figure}[!htb]
\subfigure[]{ \label{Fig.sub.711}
\includegraphics[width=0.45\textwidth]{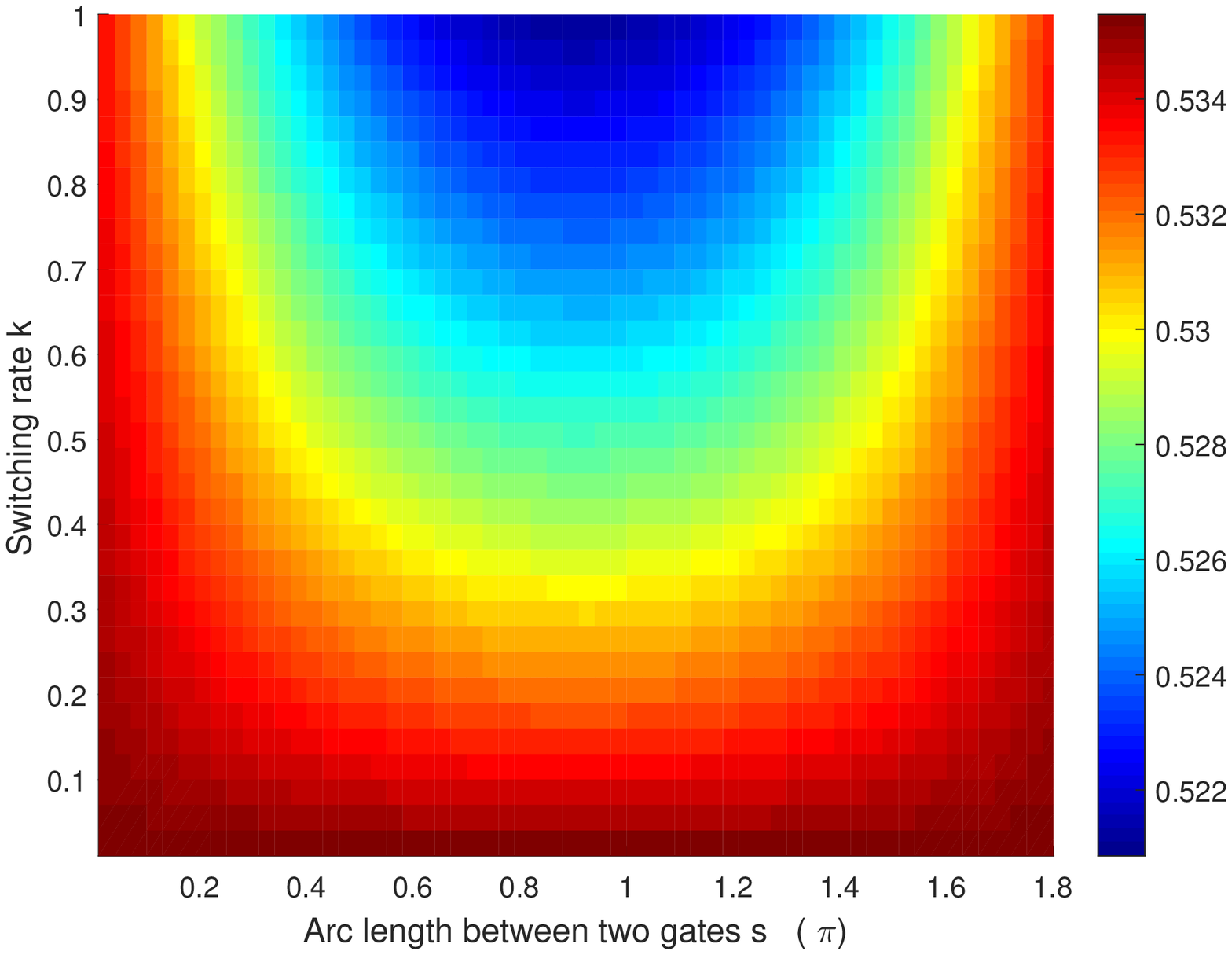}}
\subfigure[]{ \label{Fig.sub.721}
\includegraphics[width=0.45\textwidth]{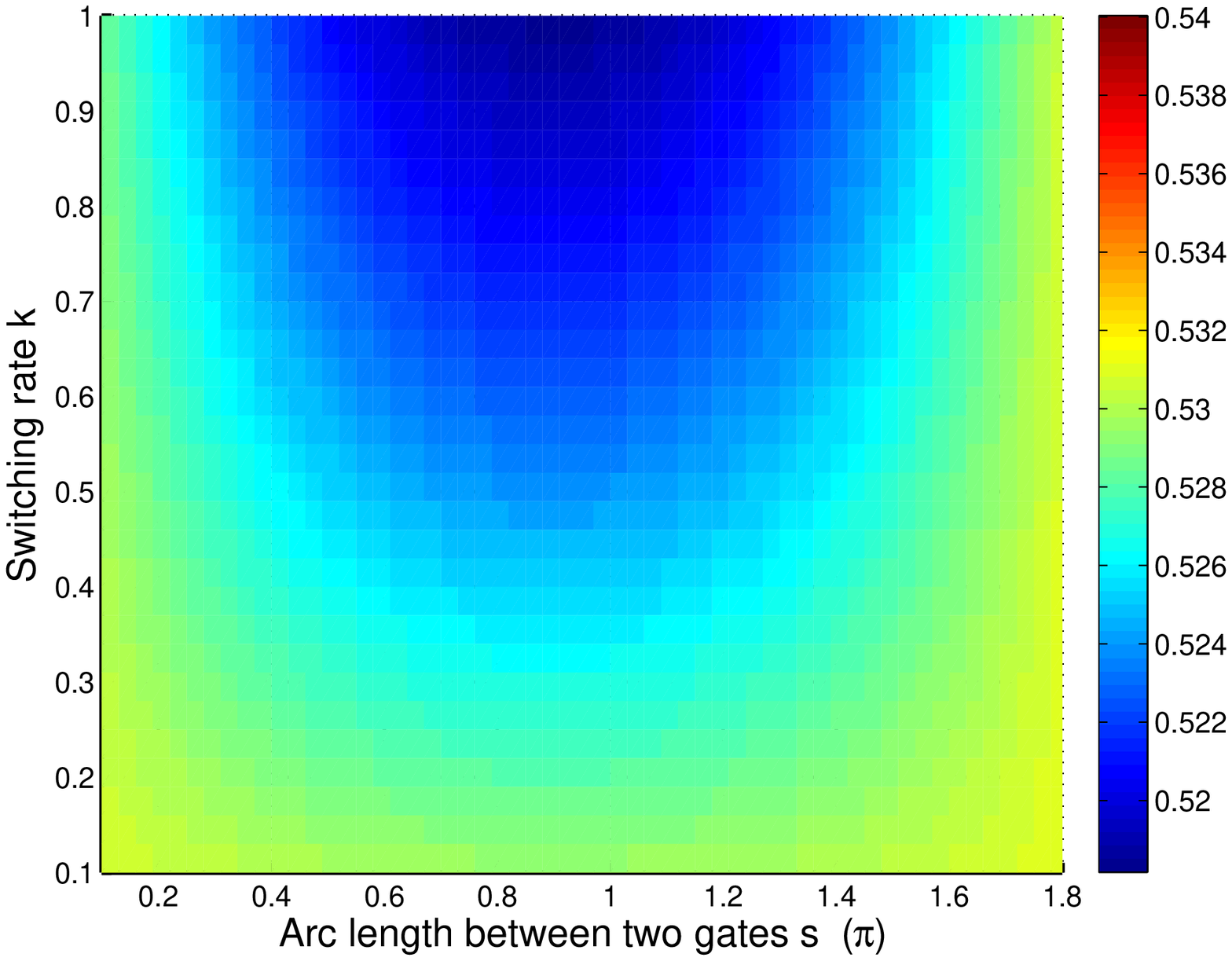}}
\caption{(Color online)  Average narrow escape time with  the arc length $s$ between two gates and switching rate $k$ between the two gates for diffusion coefficient $D=1$ and angular velocity $w=3$ : (a) With the first gate open  initially. (b) With the second gate open initially.}
 \label{Fig_71}
\end{figure}

\begin{figure}[!htb]
\subfigure[]{ \label{Fig.sub.712}
\includegraphics[width=0.45\textwidth]{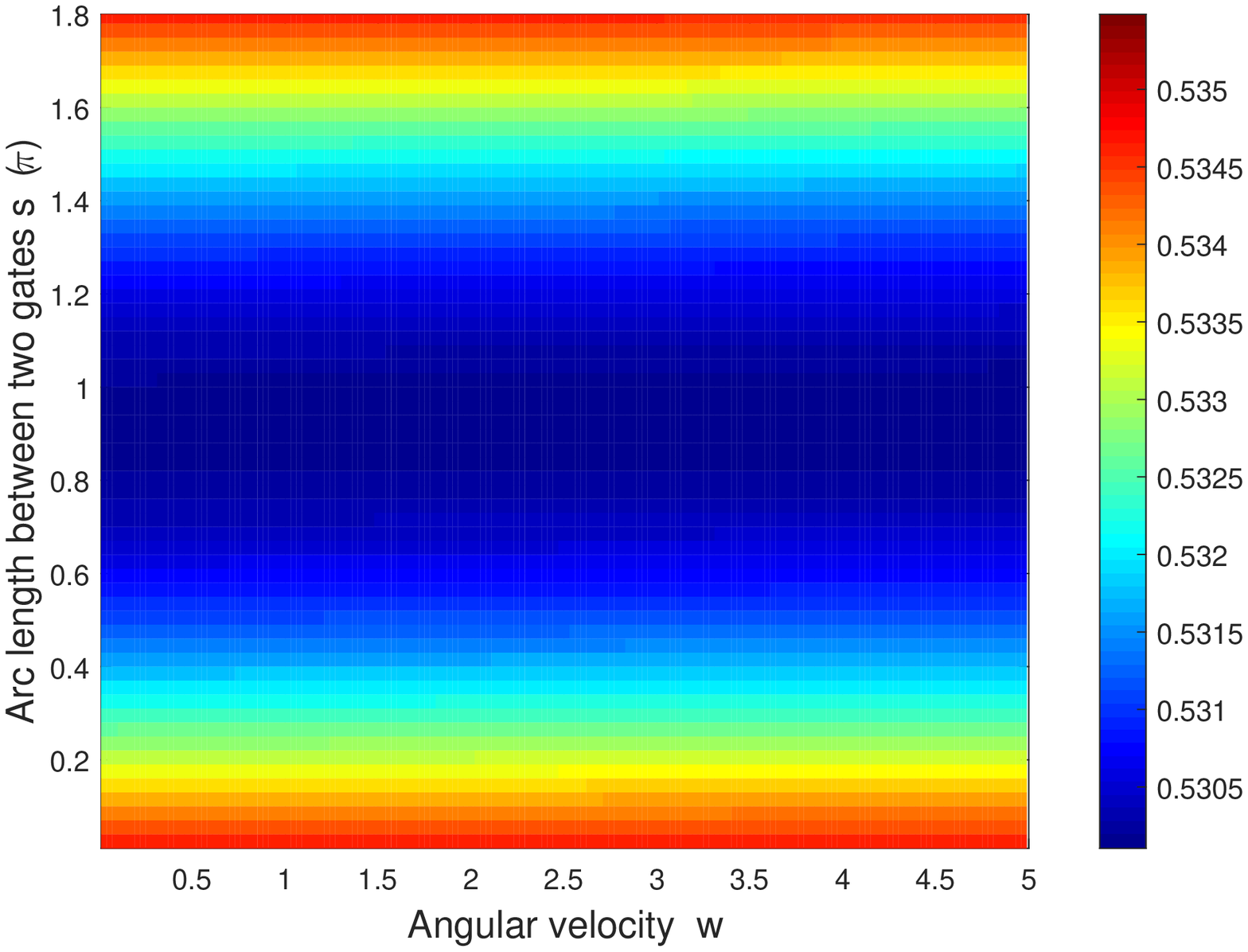}}
\subfigure[]{ \label{Fig.sub.722}
\includegraphics[width=0.45\textwidth]{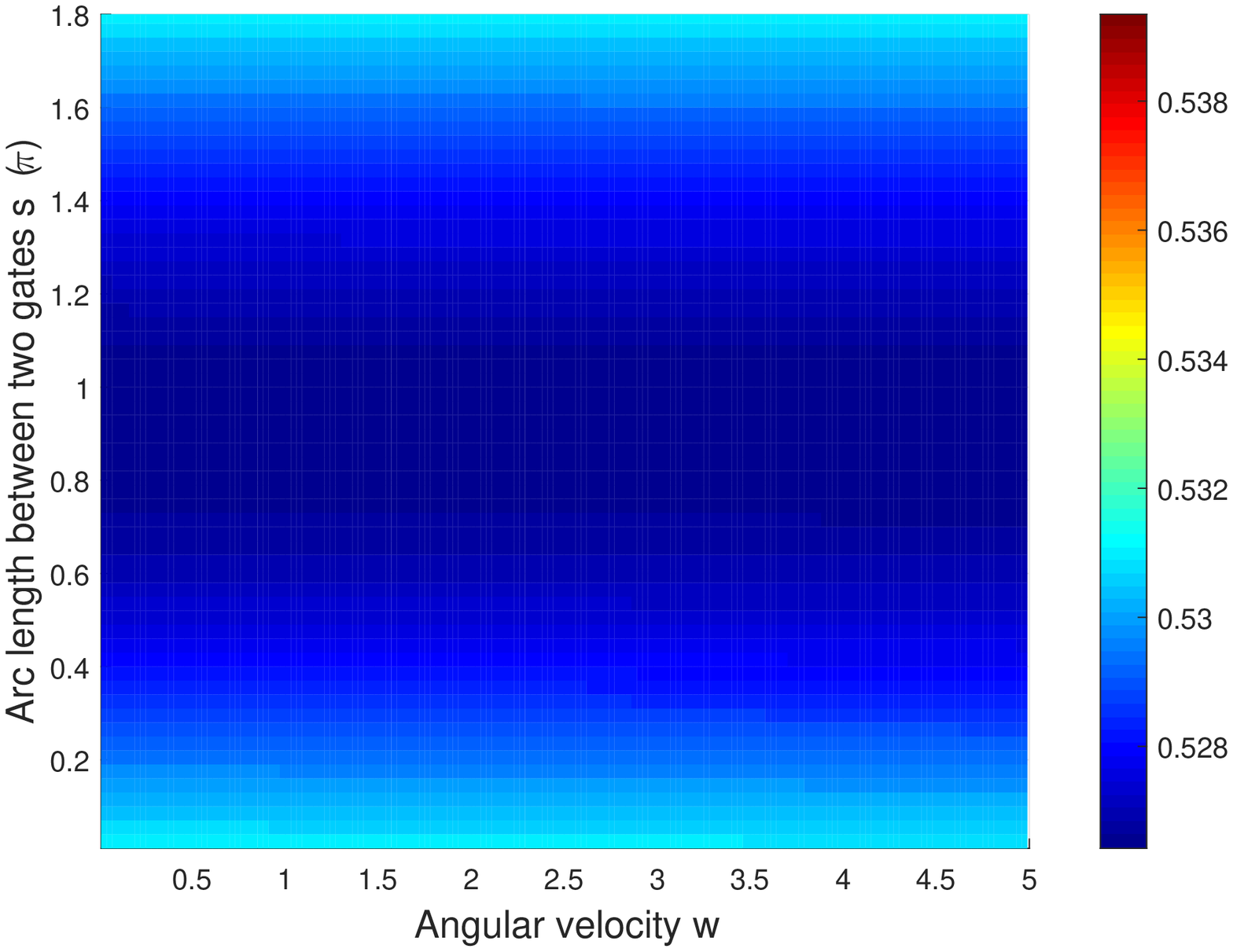}}
\caption{(Color online) Average narrow escape time with  the arc length $s$ between two gates and  angular velocity $w$ for diffusion coefficient $D=1$ and switching rate $k=0.5$ between the two gates: (a) With the first gate open  initially. (b) With the second gate open initially.}
 \label{Fig_72}
\end{figure}

In Fig \ref{Fig_71}, we display average NET as a function of arc length $s$ between two gates and switching rate $k$ between the two gates. Interestingly, we find  that the average NET $u_1$ is larger than $u_2$ in many domains. If  the initial state with the first gate open, escape from the first gate takes longer time than from the second gate. Similarly, we show the effects of arc length $s$ between two gates  and  angular velocity $w$ on  average NET in Fig \ref{Fig_72}. We find a counterintuitive phenomenon, when the initial state is that the second gate is open, escape from the second gate is faster than from the first gate.

\begin{figure}[!htb]
\subfigure[]{ \label{Fig.sub.h1}
\includegraphics[width=0.45\textwidth]{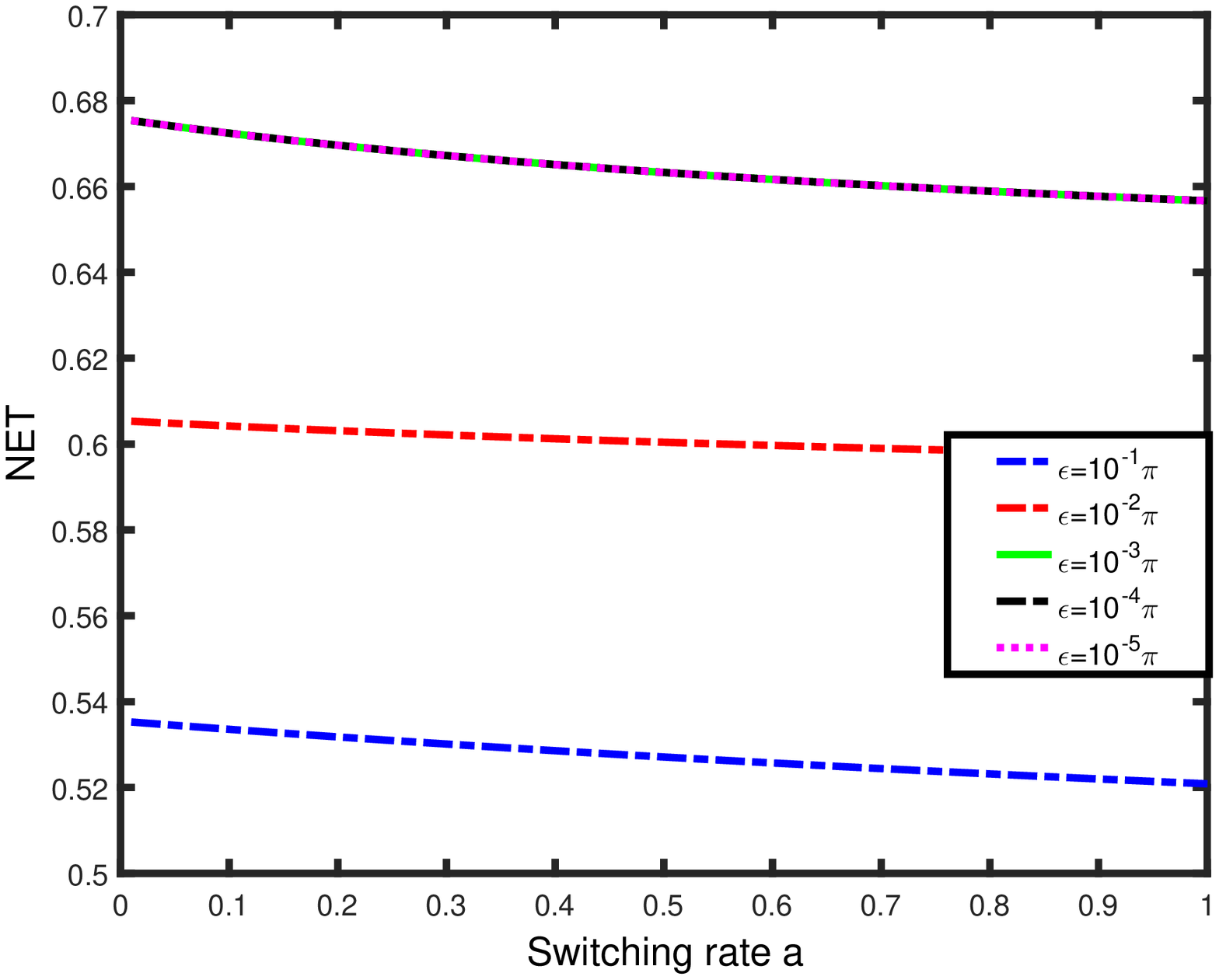}}
\subfigure[]{ \label{Fig.sub.h2}
\includegraphics[width=0.45\textwidth]{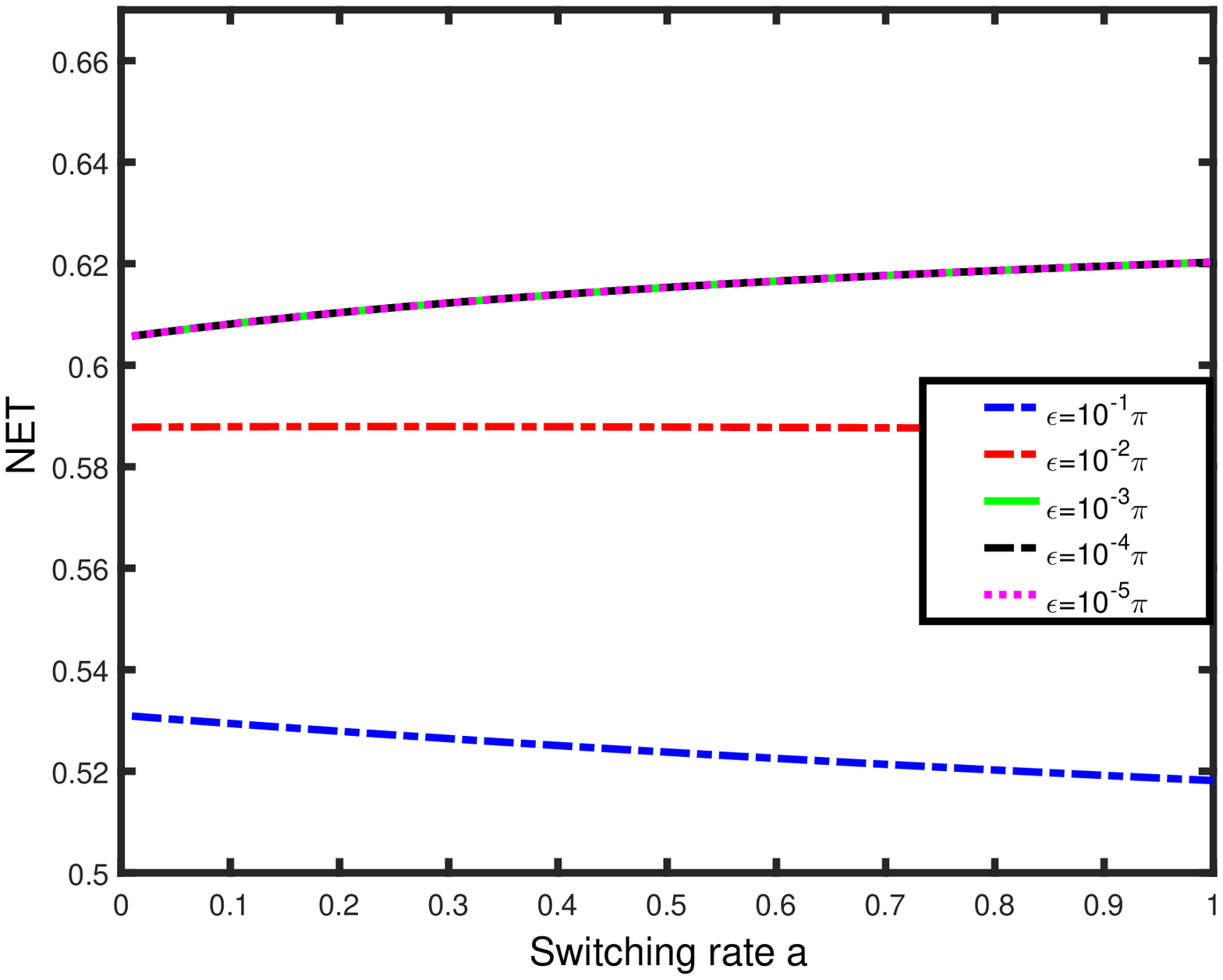}}
\caption{(Color online)  Average narrow escape time with  the switching rate $k$ between the two gates for diffusion coefficient $D=1$, angular velocity $w=3$ and arc length $s=\pi$ between two gates  for different size of gates: (a) With the first gate open initially. (b) With the second gate open initially.}
 \label{Fig h}
\end{figure}

Fig \ref{Fig h} shows the effects of the gate arc length size $\varepsilon$ on average NET. As expected, the average NET decreases with the switching rate $k$ between the two gates, but when $\varepsilon$ is sufficiently small, the average NET is independent of  gate arc length size $\varepsilon$.

\section{Conclusion}

In this work , we have  analysed the average narrow escape time through  two randomly switching gates in an annular domain, under a drift or a vector field, while most existing works  are for particles following    Brownian diffusion (i.e., no drift).  We have conducted numerical experiments to reveal    escape time's dependence on   angular velocity $w$ for the background cellular flow, switching rate $k$ between the two gates, diffusion coefficient $D$, arc length distance $s$ between two gates, and gate arc length size $\varepsilon$.


We have found that the mean escape time decreases with the switching rate $k$ between the two gates, angular velocity $w$ and diffusion coefficient $D$ for fixed arc length, but takes the minimum when the two gates are evenly separated on the  boundary for   given switching rate $k$ between the two gates. In particular, we have verified  that  when  the    arc length size $\varepsilon$ for the gates  is sufficiently small, the average narrow  escape time  is   independent of  the gate arc length size.  Our approach  selects  combinations of  system parameters (located  in the parameter space) such that the mean escape time is the longest or shortest.

The narrow escape  problems are relevant in certain  crucial processes in chemical physics and cellular biology.  Our research  sheds some lights on  the narrow escape time in a vector field,  when  system parameters vary.

\section*{Acknowledgements}
We would like to thank Xi Chen, Xiaoli Chen and  Xiaofan Li for helpful discussions for numerical schemes. This work was partly supported by the NSF, USA grant 1620449,  and the NSFC, China grants 11531006 and 11771449.

\end{document}